\documentclass[review]{elsarticle}
\usepackage[margin=2.5cm]{geometry}

\makeatletter
\def\ps@pprintTitle{%
 \let\@oddhead\@empty
 \let\@evenhead\@empty
 \def\@oddfoot{\centerline{\thepage}}%
 \let\@evenfoot\@oddfoot}
\makeatother

\usepackage{hyperref}

\usepackage{amsmath,amsfonts, mathdots, amssymb, latexsym, amsbsy,bm,amsthm,mathtools}  
\usepackage{textgreek} 
\usepackage{url}
\usepackage{enumitem}
\hypersetup{breaklinks=true}
\usepackage{cleveref}
\usepackage[ruled,vlined]{algorithm2e}
\usepackage[colorinlistoftodos]{todonotes}

\renewcommand{\mod}[1]{\;\mathrm{mod}\; #1}

\newcommand{\N}{\mathbb{N}}

\newcommand{\abs}[1]{\left\lvert #1 \right\rvert}
\newcommand{\NP}{\ensuremath{\mathsf{NP}}\xspace}

\newcommand{\EDCC}{E$\Delta$CC{} }
\renewcommand{\epsilon}{\varepsilon}
\renewcommand{\N}{\mathbb{N}}
\newtheorem{theorem}{Theorem}
\newtheorem{proposition}[theorem]{Proposition}
\newtheorem{corollary}[theorem]{Corollary}
\newtheorem{claim}[theorem]{Claim}
\newtheorem{remark}[theorem]{Remark}

\usepackage{natbib}
 \bibpunct[, ]{(}{)}{,}{a}{}{,}%

\begin{document}

\begin{frontmatter}

\title{A Greedy Algorithm for the Social Golfer and the Oberwolfach Problem}

\author[mymainaddress]{Daniel Schmand}
\ead{schmand@uni-bremen.de}

\author[mythirdaddress]{Marc Schr\"{o}der}
\ead{m.schroder@maastrichtuniversity.nl}

\author[mysecondaryaddress]{Laura Vargas Koch}
\ead{laura.vargas@oms.rwth-aachen.de}

\address[mymainaddress]{Center for Industrial Mathematics, University of Bremen, Bremen, Germany}
\address[mythirdaddress]{School of Business and Economics, Maastricht University, Maastricht, Netherlands}
\address[mysecondaryaddress]{School of Business and Economics, RWTH Aachen University, Aachen, Germany}

\begin{abstract}
Inspired by the increasing popularity of Swiss-system tournaments in sports, we study the problem of predetermining the number of rounds that can be guaranteed in a Swiss-system tournament. Matches of these tournaments are usually determined in a myopic round-based way dependent on the results of previous rounds. Together with the hard constraint that no two players meet more than once during the tournament, at some point it might become infeasible to schedule a next round. For tournaments with $n$ players and match sizes of $k\geq2$ players, we prove that we can always guarantee $\lfloor \frac{n}{k(k-1)} \rfloor$ rounds. We show that this bound is tight. This provides a simple polynomial time constant factor approximation algorithm for the social golfer problem.

We extend the results to the Oberwolfach problem. We show that a simple greedy approach guarantees at least $\lfloor \frac{n+4}{6} \rfloor$ rounds for the Oberwolfach problem. This yields a polynomial time $\frac{1}{3+\epsilon}$-approximation algorithm for any fixed $\epsilon>0$ for the Oberwolfach problem.
Assuming that El-Zahar's conjecture is true, we improve the bound on the number of rounds to be essentially tight.

\end{abstract}


\end{frontmatter}

\section{Introduction}
Swiss-system tournaments received a highly increasing consideration in the last years and are implemented in various professional and amateur tournaments in, e.g., badminton, bridge, chess, e-sports and card games. A Swiss-system tournament is a non-eliminating tournament format that features a predetermined number of rounds of competition. Assuming an even number of participants, each player plays exactly one other player in each round and two players play each other at most once in a tournament. The number of rounds is predetermined and publicly announced. The actual planning of a round usually depends on the results of the previous rounds to generate as attractive matches as possible and highly depends on the considered sport. Typically, players with the same numbers of wins in the matches played so far are paired, if possible.

Tournament designers usually agree on the fact that one should have at least $\log (n)$ rounds in a tournament with $n$ participants to ensure that there cannot be multiple players without a loss in the final rankings. \citet{appleton1995may} even mentions playing $\log (n) + 2$ rounds, so that a player may lose once and still win the tournament.

In this work, we examine a bound on the number of rounds that can be \emph{guaranteed} by tournament designers. Since the schedule of a round depends on the results of previous rounds, it might happen that at some point in the tournament, there is no next round that fulfills the constraint that two players play each other at most once in a tournament. This raises the question of how many rounds a tournament organizer can announce before the tournament starts while being sure that this number of rounds can always be scheduled. We provide bounds that are \emph{independent} of the results of the matches and the detailed rules for the setup of the rounds.

We model the feasible matches of a tournament with $n$ participants as an undirected graph with $n$ vertices. A match that is feasible in the tournament corresponds to an edge in the graph. Assuming an even number of participants, one round of the tournament corresponds to a perfect matching in the graph. After playing one round we can delete the corresponding perfect matching from the set of edges to keep track of the matches that are still feasible. We can guarantee the existence of a next round in a Swiss-system tournament if there is a perfect matching in the graph. The largest number of rounds that a tournament planner can guarantee is equal to the largest number of perfect matchings that a greedy algorithm is guaranteed to delete from the complete graph. Greedily deleting perfect matchings models the fact that rounds cannot be preplanned or adjusted later in time.

Interestingly, the results imply that infeasibility issues can arise in some state-of-the-art rules for table-tennis tournaments in Germany. There is a predefined amateur tournament series with more than 1000 tournaments per year that \emph{guarantees} the 9 to 16 participants 6 rounds in a Swiss-system tournament~\citep{httvCup}. We can show that a tournament with 10 participants might become infeasible after round 5, even if these rounds are scheduled according to the official tournament rules. Infeasible means that no matching of the players who have not played before is possible anymore and thus no rule-consistent next round. For more details, see \citet{Kuntz:Thesis:2020}. Remark~\ref{rem:extend} shows that tournament designers could \emph{extend} the lower bound from 5 to 6, by choosing the fifth round carefully.

We generalize the problem to the famous social golfer problem in which not $2$, but $k\geq 3$ players compete in each match of the tournament, see~\citet{csplib:prob010}. We still assume that each pair of players can meet at most once during the tournament. A famous example of this setting is Kirkman's schoolgirl problem \citep{kirkman1850note}, in which fifteen girls walk to school in rows of three for seven consecutive days such that no two girls walk in the same row twice.
In addition to the theoretical interest in this question, designing golf tournaments with a fixed size of the golf groups that share a hole is a common problem in the state of the art design of golfing tournaments, see e.g.,~\citet{golf}. Another application of the social golfer problem is the Volleyball Nations league. Here 16 teams play a round-robin tournament. To simplify the organisation, they repeatedly meet in groups of four at a single location and play all matches within the group. Planning which teams to group together and invite to a single location is an example of the social golfer problem. See \citet{volleyball}.

In graph-theoretic terms, a round in the social golfer problem corresponds to a set of vertex-disjoint cliques of size $k$ that contains every vertex of the graph exactly once. In graph theory, a feasible round of the social golfer problem is called a clique-factor. 
We address the question of how many rounds can be guaranteed if clique-factors, where each clique has a size of $k$, are greedily deleted from the complete graph, i.e., without any preplanning.

A closely related problem is the Oberwolfach problem. In the Oberwolfach problem, we seek to find seating assignments for multiple diners at round tables in such a way that two participants sit next to each other exactly once. Half-jokingly, we use the fact that seatings at Oberwolfach seminars are assigned greedily, and study the greedy algorithm for this problem. Instead of deleting clique-factors, the algorithm now iteratively deletes a set of vertex-disjoint cycles that contains every vertex of the graph exactly once. Such a set is called a cycle-factor. We restrict attention to the special case of the Oberwolfach problem in which all cycles have the same length $k$. We analyze how many rounds can be guaranteed if cycle-factors, in which each cycle has length $k$, are greedily deleted from the complete graph.

\subsection*{Our Contribution} Motivated by applications in sports, the social golfer problem, and the Oberwolfach problem, we study the greedy algorithm that iteratively deletes a clique, respectively cycle, factor in which all cliques/cycles have a fixed size $k$, from the complete graph. We prove the following main results for complete graphs with $n$ vertices for $n$ divisible by $k$.
\begin{itemize}
\item We can always delete $\lfloor n/(k(k-1))\rfloor$ clique-factors in which all cliques have a fixed size $k$ from the complete graph. In other words, the greedy procedure guarantees a schedule of $\lfloor n/(k(k-1))\rfloor$ rounds for the social golfer problem.
This provides a simple polynomial time $\frac{k-1}{2k^2-3k-1}$-approximation algorithm.
\item The bound of $\lfloor n/(k(k-1))\rfloor$ is tight, in the sense that it is the best possible bound we can guarantee for our greedy algorithm. To be more precise, we show that a tournament exists in which we can choose the first $\lfloor n/(k(k-1))\rfloor$ rounds in such a way that no additional feasible round exists. If a well-known conjecture by \citet{chen1994equitable} in graph theory is true (the conjecture is proven to be true for $k\leq 4$), then this is the unique example (up to symmetries) for which no additional round exists after $\lfloor n/(k(k-1))\rfloor$ rounds. In this case, we observe that for $n>k(k-1)$ we can always pick a different clique-factor in the last round such that an additional round can be scheduled.
\item We can always delete $\lfloor (n+4)/6\rfloor$ cycle-factors in which all cycles have a fixed size $k$, where $k\geq 3$, from the complete graph. This implies that our greedy approach guarantees to schedule $\lfloor (n+4)/6\rfloor$ rounds for the Oberwolfach problem. Moreover, the greedy algorithm can be implemented so that it is a polynomial time $\frac{1}{3+\epsilon}$-approximation algorithm for the Oberwolfach problem for any fixed $\epsilon>0$.
\item If El-Zahar's conjecture \citep{el1984circuits} is true (the conjecture is proven to be true for $k\leq 5$), we can increase the number of cycle-factors that can be deleted to $\lfloor (n+2)/4\rfloor$ for $k$ even and $\lfloor (n+2)/4-n/4k\rfloor$ for $k$ odd. Additionally, we show that this bound is essentially tight by distinguishing three different cases. In the first two cases, the bound is tight, i.e., an improvement would immediately disprove El-Zahar's conjecture. In the last case, a gap of one round remains.
\end{itemize}

\section{Preliminaries}

We follow standard notation in graph theory and for two graphs $G$ and $H$ we define an $H$-factor of $G$ as a union of vertex-disjoint copies of $H$ that contains every vertex of the graph $G$. 
For some graph $H$ and $n \in\mathbb{N}_{\geq 2}$, a \emph{tournament} with $r$ rounds is defined as a tuple $T=(H_1,\ldots, H_r)$ of $H$-factors of the complete graph $K_n$ such that each edge of $K_n$ is in at most one $H$-factor. The \emph{feasibility graph} of a tournament $T=(H_1,\ldots, H_r)$ is a graph $G = K_n \backslash \bigcup_{i \leq r} H_i$ that contains all edges that are in none of the $H$-factors.
If the feasibility graph of a tournament $T$ is empty, we call $T$ a \emph{complete tournament}.

Motivated by Swiss-system tournaments and the importance of greedy algorithms in real-life optimization problems, we study the greedy algorithm that starts with an empty tournament and iteratively extends the current tournament by an arbitrary $H$-factor in every round until no $H$-factor remains in the feasibility graph. We refer to Algorithm \ref{algo:greedy} for a formal description.

 \vspace{\baselineskip}
 \begin{algorithm}[H]
\SetAlgoLined
\SetKwInOut{Input}{Input}
\SetKwInOut{Output}{Output}
\Input{number of vertices $n$ and a graph $H$}
\Output{tournament $T$}
$G \leftarrow K_n$\\
$i \leftarrow 1$\\
 \While{there is an $H$-factor $H_i$ in $G$}{
    delete $H_i$ from $G$\\
    $i \leftarrow i+1$\\ 
  }
 \Return $T=(H_1,\dots,H_{i-1})$
 
\caption{Greedy tournament scheduling}
\label{algo:greedy}
\end{algorithm}
\vspace{\baselineskip}
 
\subsection*{Greedy Social Golfer Problem}

In the greedy social golfer problem we consider tournaments with $H=K_k$, for $k\geq 2$, where $K_k$ is the complete graph with $k$ vertices and all $\frac{k(k-1)}{2}$ edges. The greedy social golfer problem asks for the minimum number of rounds of a tournament computed by Algorithm \ref{algo:greedy}, as a function of $n$ and $k$. The solution of the greedy social golfer problem is a guarantee on the number of $K_k$-factors that can be iteratively deleted from the complete graph without any preplanning.
For sports tournaments this corresponds to $n$ players being assigned to rounds with matches  of size $k$ such that each player is in exactly one match per round and each pair of players meets at most once in the tournament. 

\subsection*{Greedy Oberwolfach Problem}
In the greedy Oberwolfach problem we consider tournaments with $H=C_k$, for $k\geq 3$, where $C_k$ is the cycle graph with $k$ vertices and $k$ edges. The greedy Oberwolfach problem asks for the minimum number of rounds calculated by Algorithm~\ref{algo:greedy}, given $n$ and $k$. This corresponds to a guarantee on the number of $C_k$-factors that can always be iteratively deleted from the complete graph without any preplanning.\\

Observe that for $k= 3$, both problems are equivalent.
To avoid trivial cases, we assume throughout the paper that $n$ is divisible by $k$. This is a necessary condition for the existence of a \emph{single} round. Usually, in real-life sports tournaments, additional dummy players are added to the tournament if $n$ is not divisible by $k$. The influence of dummy players on the tournament planning strongly depends on the sport. There are sports, like e.g.,\ golf or karting where matches can still be played with less than $k$ players, or others where the match needs to be cancelled if one player is missing, for example beach volleyball or tennis doubles. Thus, the definition of a best possible round if $n$ is not divisible by $k$ depends on the application. We exclude the analysis of this situation from this work to ensure a broad applicability of our results and focus on the case $n \equiv 0 \mod k$.

\subsection{Related Literature}

For matches with $k=2$ players, \cite{rosa1982premature} studied the question whether a given tournament can be extended to a round-robin tournament. This question was later solved by \cite{Csaba2016} as long as the graph is sufficiently large. They showed that even if we apply the greedy algorithm for the first $n/2-1$ rounds, the tournament can be extended to a complete tournament by choosing all subsequent rounds carefully.

\cite{cousins1975maximal} asked the question of how many rounds can be guaranteed to be played in a Swiss-system tournament for the special case $k=2$. They showed that $\frac{n}{2}$ rounds can be guaranteed. Our result of $\left\lfloor\frac{n}{k(k-1)}\right\rfloor$ rounds for the social golfer problem is a natural generalization of this result. \cite{rees1991spectrum} investigated in more detail after how many rounds a Swiss-system tournament can get stuck.

For a match size of $k\geq 2$ players, the original \emph{social golfer problem} with $n\geq 2$ players asks whether a complete tournament with $H=K_k$ exists. For $H=K_2$, such a complete tournament coincides with a round-robin tournament. Round-robin tournaments are known to exist for every even number of players. Algorithms to calculate such schedules are known for more than a century due to \citet{schurig1886}. For a more recent survey on round-robin tournaments, we refer to \cite{rasmussen2008round}.

For $H=K_k$ and $k\geq 2$, complete tournaments are also known as resolvable balanced incomplete block designs (resolvable-BIBDs). To be precise, a \emph{resolvable-BIBD} with parameters $(n,k,1)$ is a collection of subsets (blocks) of a finite set $V$ with $|V|=n$ elements with the following properties:
\begin{enumerate}
\item Every pair of distinct elements $u,v$ from $V$ is contained in exactly one block.
\item Every block contains exactly $k$ elements.
\item The blocks can be partitioned into rounds $R_1, R_2, \ldots , R_r$ such that each element of $V$ is contained in exactly one block of each round. 
\end{enumerate}

Notice that a round in a resolvable-BIBD corresponds to an $H$-factor in the social golfer problem.
Similar to the original social golfer problem, a resolvable-BIBD consists of $(n-1)/(k-1)$ rounds. For the existence of a resolvable-BIBD the conditions $n \equiv 0 \mod{k}$ and $n-1 \equiv 0 \mod{k-1}$ are clearly necessary. For $k=3$, \citet{ray1971solution} proved that these two conditions are also sufficient. Later, \citet{hanani1972resolvable} proved the same result for $k=4$. In general, these two conditions are not sufficient (one of the smallest exceptions being $n=45$ and $k=5$), but \citet{ray1973existence} showed that they are asymptotically sufficient, i.e., for every $k$ there exist a constant $c(k)$ such that the two conditions are sufficient for every $n$ larger than $c(k)$. These results immediately carry over to the existence of a \emph{complete} tournament with $n$ players and $H=K_k$.

Closely related to our problem is the question of the existence of graph factorizations. An $H$-factorization of a graph $G$ is collection of $H$-factors that exactly cover the whole graph $G$. For an overview of graph theoretic results, we refer to \cite{yuster2007combinatorial}. \cite{condon2019bandwidth} looked at the problem of maximizing the number of $H$-factors when choosing rounds carefully. For our setting, their results imply that in a sufficiently large graph one can always schedule rounds such that the number of edges remaining in the feasibility graph is an arbitrary small fraction of edges. Notice that the result above assumes that we are able to preplan the whole tournament. In contrast, we plan rounds of a tournament in an online fashion depending on the results in previous rounds.

In greedy tournament scheduling, Algorithm~\ref{algo:greedy} greedily adds one round after another to the tournament, and thus \emph{extends} a given tournament step by step. The study of the existence of another feasible round in a given tournament with $H=K_k$ is related to the existence of an equitable graph-coloring. Given some graph $G=(V,E)$, an $\ell$-coloring
is a function $f: V \rightarrow \{1, \ldots, \ell \}$, such that $f(u) \neq f(v)$ for all edges $(u,v) \in E$. An \emph{equitable $\ell$-coloring} is an $\ell$-coloring, where the number of vertices in any two color classes differs by at most one, i.e., $|\{ v | f(v)=i \}| \in \{ \left\lfloor \frac{n}{\ell} \right\rfloor , \left\lceil \frac{n}{\ell} \right\rceil\}$ for every color $i \in \{1, \ldots , \ell \}$.
 
To relate equitable colorings of graphs to the study of the extendability of tournaments, we consider the complement graph $\bar{G}$ of the feasibility graph $G=(V,E)$, as defined by
$\bar{G}=K_n \backslash E$. Notice that a color class in an equitable coloring of the vertices of $\bar{G}$ is equivalent to a clique in $G$. In an equitable coloring of $\bar{G}$ with $\frac{n}{k}$ colors, each color class has the same size, which is $k$. Thus, finding an equitable $\frac{n}{k}$-coloring in $\bar{G}$ is equivalent to finding a $K_k$-factor in $G$ and thus an extension of the tournament. Questions on the existence of an equitable coloring dependent on the vertex degrees in a graph have already been considered by \citet{erdos1964problem}, who posed a conjecture on the existence of equitable colorings in low degree graphs, that has been proven by \citet{hajnal1970proof}. Their proof was simplified by \citet{kierstead2010fast}, who also gave a polynomial time algorithm to find an equitable coloring. In general graphs, the existence of clique-factors with clique size equal to $3$ \citep[][Sec.~3.1.2]{garey1979computers} and at least $3$ \citep{kirkpatrick1978completeness,kirkpatrick1983complexity,hell1984packings} is known to be NP-hard.
 
 The maximization variant of the social golfer problem for $n$ players and $H=K_k$ asks for a schedule which lasts as many rounds as possible. It is mainly studied in the constraint programming community using heuristic approaches \citep{dotu2005scheduling, triska2012effective,triska2012improved, liu2019social}. Our results give lower bounds for the maximization variant using a very simple greedy algorithm.

For $n$ players and table sizes $k_1, \ldots, k_{\ell}$ with $n=k_1 + \ldots +k_{\ell}$, the (classical) \emph{Oberwolfach problem} can be stated as follows. Defining $\tilde{H} = \bigcup_{i\leq \ell} C_{k_i}$ the problem asks for the existence of a tournament of $n$ players, with $H=\tilde{H}$ which has $(n-1)/2$ rounds. Note that the Oberwolfach problem does not ask for such an assignment but only for existence. While the general problem is still open, several special cases have been solved. Assuming $k=k_1=\ldots=k_{\ell}$, \citet{alspach1989oberwolfach} showed existence for all odd $k$ and all $n$ odd with $n \equiv 0 \mod{k}$. For $k$ even, \citet{alspach1989oberwolfach} and \citet{hoffman1991existence} analyzed a slight modification of the Oberwolfach problem and showed that there is a tournament, such that the corresponding feasibility graph $G$ is not empty, but equal to a perfect matching for all even $n$ with $n \equiv 0 \mod{k}$. 
Recently, the Oberwolfach problem was solved for large $n$, see \cite{glock2021resolution}, and for small $n$, see \cite{salassa}.

\citet{liu2003equipartite} studied a variant of the Oberwolfach problem in bipartite graphs and gave conditions under which the existence of a complete tournament is guaranteed.

A different optimization problem inspired by finding feasible seating arrangements subject to some constraints is given by \cite{estanislaomeunier}.

The question of extendability of a given tournament with $H=C_k$ corresponds to the covering of the feasibility graph with cycles of length $k$. Covering graphs with cycles is already studied since \citet{petersen1891theorie}. The problem of finding a set of cycles of arbitrary lengths covering a graph (if one exists) is polynomially solvable \citep{edmonds1970matching}. However, if certain cycle lengths are forbidden, the problem is NP-complete \citep{hell1988restricted}.

\subsection{Example}
Consider the example of a tournament with $n=6$ and $H=K_2$ depicted in Figure \ref{fig:exa}. The coloring of the edges in the graph on the left represents three rounds $H_1, H_2, H_3$. The first round $H_1$ is depicted by the set of red edges. Each edge corresponds to a match. In the second round, all blue edges are played. The third round $H_3$ consists of all green edges. After these three rounds, the feasibility graph $G$ of the tournament is depicted on the right side of the figure. We cannot feasibly schedule a next round as there is no perfect matching in $G$. Equivalently, we can observe that the tournament with $3$ rounds cannot be extended, since there is no equitable $3$-coloring in $\bar{G}$, which is depicted on the left of Figure~\ref{fig:exa}.

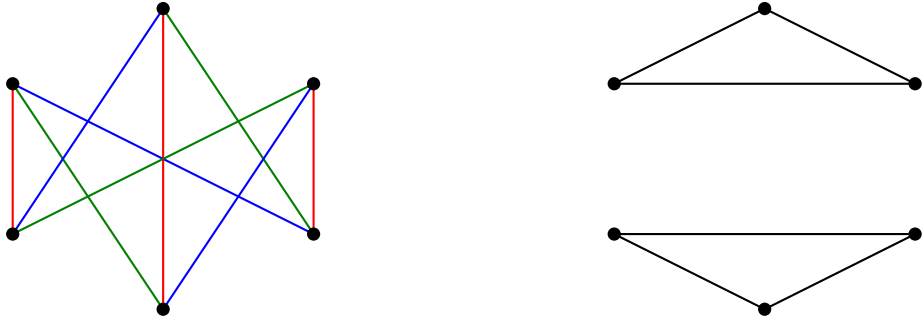
\begin{figure}[t]
\begin{minipage}{0.55\textwidth}
\centering
\begin{tikzpicture}
    \draw[thick,red] (-2,1) -- (-2,-1);
    \draw[thick,green!50!black] (-2,1) -- (0,-2);
    \draw[thick,blue] (-2,1) -- (2,-1);
    \draw[thick,blue] (0,2) -- (-2,-1);
    \draw[thick,red] (0,2) -- (0,-2);
    \draw[thick,green!50!black] (0,2) -- (2,-1);
    \draw[thick,green!50!black](2,1) -- (-2,-1);
    \draw[thick,blue] (2,1) -- (0,-2);
    \draw[thick,red] (2,1) -- (2,-1);
    \node[circle,fill=black,inner sep=0pt,minimum size=5pt] at (-2,1){};
    \node[circle,fill=black,inner sep=0pt,minimum size=5pt] at (0,2){};
    \node[circle,fill=black,inner sep=0pt,minimum size=5pt] at (2,1){};
    \node[circle,fill=black,inner sep=0pt,minimum size=5pt] at (-2,-1){};
    \node[circle,fill=black,inner sep=0pt,minimum size=5pt] at (0,-2){};
    \node[circle,fill=black,inner sep=0pt,minimum size=5pt] at (2,-1){};
    \node at (0,-2.5) {};
\end{tikzpicture}
\end{minipage}
\begin{minipage}{0.4\textwidth}
\centering
\begin{tikzpicture}
    \draw[thick] (-2,1) -- (0,2);
    \draw[thick] (-2,1) -- (2,1);
    \draw[thick] (0,2) -- (2,1);
    \draw[thick] (-2,-1) -- (0,-2);
    \draw[thick] (-2,-1) -- (2,-1);
    \draw[thick] (0,-2) -- (2,-1);
    \node[circle,fill=black,inner sep=0pt,minimum size=5pt] at (-2,1){};
    \node[circle,fill=black,inner sep=0pt,minimum size=5pt] at (0,2){};
    \node[circle,fill=black,inner sep=0pt,minimum size=5pt] at (2,1){};
    \node[circle,fill=black,inner sep=0pt,minimum size=5pt] at (-2,-1){};
    \node[circle,fill=black,inner sep=0pt,minimum size=5pt] at (0,-2){};
    \node[circle,fill=black,inner sep=0pt,minimum size=5pt] at (2,-1){};
    \node at (0,-2.5){};
\end{tikzpicture}
\end{minipage}

\caption{Consider a tournament with 6 participants and $H=K_2$. The left figure corresponds to three rounds, where each color denotes the matches of one round. The right figure depicts the feasibility graph after these three rounds.}\label{fig:exa}
\end{figure}

On the other hand there is a tournament with $n=6$ and $H=K_2$ that consists of $5$ rounds. The corresponding graph is depicted in Figure~\ref{fig:com}. Since this is a complete tournament, the example is a resolvable-BIBD with parameters $(6,2,1)$. The vertices of the graph correspond to the finite set $V$ of the BIBD and the colors in the figure correspond to the rounds in the BIBD. Note that these examples show that there is a complete tournament with $n=6$ and $H=K_2$, where $5$ rounds are played while the greedy algorithm can get stuck after $3$ rounds. In the remainder of the paper, we aim for best possible bounds on the number of rounds that can be guaranteed by using the greedy algorithm.

 \begin{figure}[t]
\centering
\begin{tikzpicture}
    \draw[thick,red] (-2,1) -- (0,2);
    \draw[thick] (-2,1) -- (2,1);
    \draw[very thick,yellow!90!black] (0,2) -- (2,1);
    \draw[thick,red] (-2,-1) -- (0,-2);
    \draw[thick] (-2,-1) -- (2,-1);
    \draw[very thick,yellow!90!black] (0,-2) -- (2,-1);
    \draw[very thick,yellow!90!black] (-2,1) -- (-2,-1);
    \draw[thick,green!50!black] (-2,1) -- (0,-2);
    \draw[thick,blue] (-2,1) -- (2,-1);
    \draw[thick,blue] (0,2) -- (-2,-1);
    \draw[thick] (0,2) -- (0,-2);
    \draw[thick,green!50!black] (0,2) -- (2,-1);
    \draw[thick,green!50!black](2,1) -- (-2,-1);
    \draw[thick,blue] (2,1) -- (0,-2);
    \draw[thick,red] (2,1) -- (2,-1);
    \node[circle,fill=black,inner sep=0pt,minimum size=5pt] at (-2,1){};
    \node[circle,fill=black,inner sep=0pt,minimum size=5pt] at (0,2){};
    \node[circle,fill=black,inner sep=0pt,minimum size=5pt] at (2,1){};
    \node[circle,fill=black,inner sep=0pt,minimum size=5pt] at (-2,-1){};
    \node[circle,fill=black,inner sep=0pt,minimum size=5pt] at (0,-2){};
    \node[circle,fill=black,inner sep=0pt,minimum size=5pt] at (2,-1){};
\end{tikzpicture}
\caption{A complete tournament with 6 players and 5 rounds, in which each color represents the matches of a round.}\label{fig:com}
\end{figure}
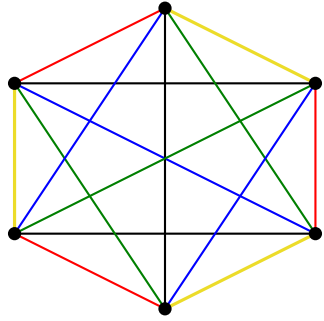

\subsection{Outline}
The paper is structured as follows. We start with the analysis of Swiss-system tournaments to demonstrate our main ideas. To be more precise, Section \ref{sec:war} considers the setting of greedy tournament scheduling with $H=K_2$. Section \ref{sec:gsgp} then generalizes the main results for the greedy social golfer problem. Lastly, in Section \ref{sec:gop}, we obtain lower and upper bounds on the number of rounds for the greedy Oberwolfach problem.

\section{Warmup: Perfect Matchings}\label{sec:war}

Most sports tournaments consist of matches between two competing players. We therefore first consider the special case 
of a tournament with $H=K_2$. 
In this setting, the greedy social golfer problem boils down to iteratively deleting perfect matchings from the complete graph. Recall that Propositions \ref{prop:k=2} and \ref{prop:k=2l}, and Corollary \ref{corr:ndivbyfour} were already shown by \cite{cousins1975maximal}. For completeness, we have added the proofs.

First, we use Dirac's theorem to show that we can always greedily delete at least $\frac{n}{2}$ perfect matchings from the complete graph. Recall that we assume $n$ to be even to guarantee the existence of a single perfect matching.

\begin{proposition}\label{prop:k=2}
For each even $n\in\mathbb{N}$ and $H=K_2$, Algorithm~\ref{algo:greedy} outputs a tournament with at least $\frac{n}{2}$ rounds.
\end{proposition}

\begin{proof}
Algorithm~\ref{algo:greedy} starts with an empty tournament and extends it by one round in every iteration.
To show that Algorithm~\ref{algo:greedy} runs for at least $\frac{n}{2}$ iterations we consider the feasibility graph of the corresponding tournament. Recall that the degree of each vertex in a complete graph with $n$ vertices is $n-1$. In each round, the algorithm deletes a perfect matching and thus the degree of a vertex is decreased by $1$. After at most $\frac{n}{2}-1$ rounds, the degree of every vertex is at least $\frac{n}{2}$. By Dirac's theorem \citep{dirac1952some}, a Hamiltonian cycle exists. The existence of a Hamiltonian cycle implies the existence of a perfect matching by taking every second edge of the Hamiltonian cycle. So after at most $\frac{n}{2}-1$ rounds, the tournament can be extended and the algorithm does not terminate.
\end{proof}

Second, we prove that the bound of Proposition \ref{prop:k=2} is tight by showing that there are tournaments that cannot be extended after $\frac{n}{2}$ rounds.
\begin{proposition}\label{prop:k=2l}
There are infinitely many $n \in \N$ for which there exists a tournament that cannot be extended after $\frac{n}{2}$ rounds.
\end{proposition}
\begin{proof}
Choose $n$ such that $\frac{n}{2}$ is odd. We describe the chosen tournament by perfect matchings in the feasibility graph $G$. Given a complete graph with $n$ vertices, we partition the vertices into a set $A$ with $|A|=\frac{n}{2}$ and $V\setminus A$ with $|V\setminus A|=\frac{n}{2}$. We denote the players in $A$ by $1,\ldots, \frac{n}{2}$ and the players in $V\setminus A$ by $\frac{n}{2}+1,\ldots,n$.

In each round $r=1,\ldots,\frac{n}{2}$, player $i+\frac{n}{2}$ is scheduled in a match with player $i+r-1$ (modulo $\frac{n}{2}$) for all $i=1,\ldots,\frac{n}{2}$. After deleting these $\frac{n}{2}$ perfect matchings, the feasibility graph $G$ consists of two disjoint complete graphs of size $\frac{n}{2}$, as every player in $A$ has played against every player in $V\setminus A$. Given that $\frac{n}{2}$ is odd, no perfect matching exists and hence the tournament cannot be extended.
\end{proof}

A natural follow-up question is to characterize those feasibility graphs that can be extended after $\frac{n}{2}$ rounds. Proposition \ref{prop:cha} answers this question and we essentially show that the provided example is the only graph structure that cannot be extended after $\frac{n}{2}$ rounds.

\begin{proposition}\label{prop:cha}
Let $T$ be a tournament of $\frac{n}{2}$ rounds with feasibility graph $G$ and its complement $\bar{G}$. Then $T$ cannot be extended if and only if $\bar{G} = K_{\frac{n}{2},\frac{n}{2}}$ and $\frac{n}{2}$ is odd.
\end{proposition}

Before we prove the proposition we present a result by \citet{chen1994equitable}, which the proof makes use of.

\subsubsection*{Chen-Lih-Wu theorem \citep{chen1994equitable}.}
Let $G$ be a connected graph with maximum degree $\Delta(G) \geq \frac{n}{2}$. If $G$ is different from $K_m$ and $K_{2m+1,2m+1}$ for all $m\geq 1$, then $G$ is equitable $\Delta(G)$-colorable.

\begin{proof}[Proof of Proposition \ref{prop:cha}.]
If the complement of the feasibility graph $\bar{G} = K_{\frac{n}{2},\frac{n}{2}}$ with $\frac{n}{2}$ odd, we are exactly in the situation of the proof of Proposition~\ref{prop:k=2l}. To show equivalence, assume that either $\bar{G} \neq K_{\frac{n}{2},\frac{n}{2}}$ or $\frac{n}{2}$ even.
By using the Chen-Lih-Wu Theorem, we show that in this case $\bar{G}$ is equitable $\frac{n}{2}$-colorable.

After $\frac{n}{2}$ rounds, we have $\Delta(\bar{G})=\frac{n}{2}$. We observe that $\bar{G}= K_n$ if and only if $n=2$ and in this case $\bar{G} = K_{1,1}$, a contradiction.  Thus all conditions of the Chen-Lih-Wu theorem are fulfilled, and $\bar{G}$ is equitable $\frac{n}{2}$-colorable. An equitable $\frac{n}{2}$-coloring in $\bar{G}$ corresponds to a perfect matching in $G$ and hence implies that the tournament is extendable.
\end{proof}

\begin{corollary}
\label{corr:ndivbyfour}
For each even $n \in \mathbb{N}$ divisible by four and $H=K_2$,  Algorithm~\ref{algo:greedy} outputs a tournament with at least $\frac{n}{2}+1$ rounds.
\end{corollary}

\begin{remark}
\label{rem:extend}
By selecting the perfect matching in round $\frac{n}{2}$ carefully, there always exists a tournament with $\frac{n}{2}+1$ rounds.
\end{remark}
\begin{proof}
After $\frac{n}{2}-1$ rounds of a tournament $T$, the degree of every vertex in $G$ is at least $\frac{n}{2}$. By Dirac's theorem \citep{dirac1952some}, there is a Hamiltonian cycle in $G$. This implies that two edge-disjoint perfect matchings exist: one that takes every even edge of the Hamiltonian cycle and one that takes every odd edge of the Hamiltonian cycle. If we first extend $T$ by taking every even edge of the Hamiltonian cycle and then extend $T$ by taking every odd edge of the Hamiltonian cycle, we have a tournament of $\frac{n}{2}+1$ rounds.
\end{proof}

\section{The Greedy Social Golfer Problem}\label{sec:gsgp}
We generalize the previous results to $k\geq 3$. This means we analyze tournaments with $n$ participants and $H=K_k$. Dependent on $n$ and $k$, we provide tight bounds on the number of rounds that can be scheduled greedily, i.e., by using Algorithm~\ref{algo:greedy}. 
Remember that we assume that $n$ is divisible by $k$.

\begin{theorem}\label{thm:k>2}
For each $n \in \mathbb{N}$ and $H= K_k$, Algorithm~\ref{algo:greedy} outputs a tournament with at least $\lfloor\frac{n}{k(k-1)}\rfloor$ rounds.
\end{theorem}

Before we continue with the proof, we first state a result from graph theory. In our proof, we will use the Hajnal-Szemeredi theorem and adapt it such that it applies to our setting.

\subsubsection*{Hajnal-Szemeredi Theorem \citep{hajnal1970proof}.}
Let $G$ be a graph with $n\in\N$ vertices and maximum vertex degree $\Delta(G)\leq \ell-1$. Then $G$ is equitable $\ell$-colorable.

\begin{proof}[Proof of Theorem \ref{thm:k>2}.]
We start by proving the lower bound on the number of rounds. Assume for sake of contradiction that there are $n \in \mathbb{N}$ and $k \in \mathbb{N}$ such that the greedy algorithm for $H=K_k$ terminates with a tournament $T$ with  $r\leq \lfloor\frac{n}{k(k-1)}\rfloor -1$ rounds. We will use the feasibility graph $G$ corresponding to $T$. Recall that the degree of a vertex in a complete graph with $n$ vertices is $n-1$. For each $K_k$-factor $(H_1, \dots, H_r)$, every vertex loses $k-1$ edges. Thus, every vertex in $G$ has degree
\[n-1 - r(k-1) \geq  n-1-\left(\Big\lfloor\frac{n}{k(k-1)}\Big\rfloor -1\right)(k-1) \geq n-1-\frac{n}{k}+k-1\;.\]
We observe that each vertex in the complement graph $\bar{G}$ has at most degree $\frac{n}{k} - k +1$. Using the Hajnal-Szemeredi theorem  with $\ell = \frac{n}{k}$, we obtain the existence of an $\frac{n}{k}$-coloring where all color classes have size $k$. Since there are no edges between vertices of the same color class in $\bar{G}$, they form a clique in $G$. Thus, there exists a $K_k$-factor in $G$, which is a contradiction to the assumption that Algorithm \ref{algo:greedy} terminated.
 This implies that $r>\lfloor\frac{n}{k(k-1)}\rfloor -1$, i.e., the total number of rounds is at least $\lfloor\frac{n}{k(k-1)}\rfloor$.
\end{proof}

\begin{remark}
\citet{kierstead2010fast} showed that finding a clique-factor can be done in polynomial time if the minimum vertex degree is at least $\frac{n(k-1)}{k}$.
\end{remark}

Let OPT be the maximum possible number of rounds of a tournament. We conclude that the greedy algorithm is a constant factor approximation factor for the social golfer problem.

\begin{corollary}
Algorithm~\ref{algo:greedy} outputs at least $\frac{1}{k}\text{OPT}-1$ rounds for the social golfer problem. Thus it is a $\frac{k-1}{2k^2-3k-1}$-approximation algorithm for the social golfer problem.
\end{corollary}

\begin{proof}
The first statement follows directly from Theorem \ref{thm:k>2} and the fact that $\text{OPT} \leq \frac{n-1}{k-1}$.

For proving the second statement, we first consider the case $n\leq 2k(k-1)-k$. Note that the algorithm always outputs at least one round. OPT is upper bounded by $\frac{n-1}{k-1} \leq \frac{2k (k-1)- k-1}{k-1}=\frac{2k^2-3k-1}{k-1}$, which implies the approximation factor. 
For $n\geq 2k(k-1)-k$, observe that the greedy algorithm guarantees to output $\lfloor \frac{n}{k(k-1)} \rfloor$ rounds in polynomial time by Theorem \ref{thm:k>2}. This yields
\begin{align*}
    \frac{\left \lfloor \frac{n}{k(k-1)}\right\rfloor}{\frac{n-1}{k-1}} &\geq   \frac{\frac{n - \left(k (k-1)-k\right)}{k(k-1)}}{\frac{n-1}{k-1}}\geq \frac{\frac{2k(k-1)-k - k (k-1)+k}{k(k-1)}}{\frac{2k(k-1)-k-1}{k-1}}=\frac{k-1}{2k^2-3k-1},
\end{align*}
where the first inequality follows since we round down by at most $\frac{k(k-1)-k}{k(k-1)}$ and the second inequality follows since the second term is increasing in $n$ for $k \geq 3$.
\end{proof}

Our second main result on greedy tournament scheduling with $H=K_k$ shows that the bound of Theorem \ref{thm:k>2} is tight.

\begin{theorem}
There are infinitely many $n \in \N$ for which there exists a tournament that cannot be extended after $\lfloor\frac{n}{k(k-1)}\rfloor$ rounds.
\label{lowerboundexample_k>2}
\end{theorem}

\begin{proof}
We construct a tournament with $n=j(k(k-1))$ participants for some $j$ large enough to be chosen later. We will define necessary properties of $j$ throughout the proof and argue in the end that there are infinitely many possible integral choices for $j$. The tournament we will construct has $\lfloor\frac{n}{k(k-1)}\rfloor$ rounds and we will show that it cannot be extended. Note that $\lfloor\frac{n}{k(k-1)}\rfloor = \frac{n}{k(k-1)}$.

The proof is based on a step-by-step modification of the feasibility graph $G$. We will start with the complete graph $K_n$ and describe how to delete $\frac{n}{k(k-1)}$ $K_k$-factors such that the resulting graph does not contain a $K_k$-factor. This is equivalent to constructing a tournament with $\lfloor\frac{n}{k(k-1)}\rfloor$ rounds that cannot be extended.

Given a complete graph with $n$ vertices, we partition the vertices $V$ in two sets, a set $A$ with $\ell=\frac{n}{k}+1$ vertices and a set $V \backslash A$ with $n-\ell$ vertices. We will choose all $\frac{n}{k(k-1)}$ $K_k$-factors in such a way, that no edge $\{a,b\}$ with $a\in A$ and $b\notin A$ is deleted, i.e., each $K_k$ is either entirely in $A$ or entirely in $V\setminus A$. We will explain below that this is possible. Since a vertex in $A$ has $\frac{n}{k}$ neighbours in $A$ and $k-1$ of them are deleted in every $K_k$-factor, all edges within $A$ are deleted after deleting $\frac{n}{k(k-1)}$ $K_k$-factors.

We now first argue that after deleting these $\frac{n}{k(k-1)}$ $K_k$-factors, no $K_k$-factor exists. Assume that there exists another $K_k$-factor. In this case, each vertex in $A$ forms a clique with $k-1$ vertices of $V \backslash A$. However, since $(k-1)\cdot(\frac{n}{k}+1)>\frac{(k-1)n}{k}-1=|V \backslash A|$ there are not enough vertices in $V \backslash A$, a contradiction to the existence of the $K_k$-factor.

It remains to show that there are $\frac{n}{k(k-1)}$ $K_k$-factors that do not contain an edge $\{a,b\}$ with $a\in A$ and $b \notin A$. We start by showing that $\frac{n}{k(k-1)}$ $K_k$-factors can be found within $A$. \citet{ray1973existence} showed that given $k'\geq 2$ there exists a constant $c(k')$ such that if $n'\geq c(k')$ and $n' \equiv k' \mod k'(k'-1)$, then a resolvable-BIBD with parameters $(n',k',1)$ exists.

By choosing $k'=k$ and $n' = \ell$ with $j= \lambda \cdot k +1$ for some $\lambda \in \N$ large enough, we establish $\ell\geq c(k)$, where $c(k)$ is defined by \citet{ray1973existence}, and we get
\[|A| = \ell = \frac{n}{k}+1 = j(k-1)+1 = (\lambda k +1)(k-1) + 1 = k+ \lambda k (k-1)\;.\]
Thus, a resolvable-BIBD with parameters $(\ell,k,1)$ exists, and there is a complete tournament for $\ell$ players with $H=K_k$, i.e., we can find $\frac{n}{k(k-1)}$ $K_k$-factors in $A$. 

It remains to show that we also find $\frac{n}{k(k-1)}$ $K_k$-factors in $V \setminus A$. We define a tournament that we call shifting tournament as follows. We arbitrarily write the names of the players in $V \setminus A$ into a table of size $k\times (n-\ell)/k$. Each column of the table corresponds to a $K_k$ and the table to a $K_k$-factor in $V \setminus A$. By rearranging the players we get a sequence of tables, each corresponding to a $K_k$-factor. To construct the next table from a preceding one, for each row $i$, all players move $i-1$ steps to the right (modulo $(n-\ell)/k$).

We claim that this procedure results in $\frac{n}{k(k-1)}$ $K_k$-factors that do not share an edge. First, notice that the step difference between any two players in two rows $i \neq i'$ is at most $k-1$, where we have equality for rows $1$ and $k$. However, we observe that $(n-\ell)/k$ is not divisible by $(k-1)$ since $n/k$ is divisible by $k-1$ by definition, whereas $\ell/k$ is not divisible by $k-1$ since $\ell/k(k-1)=1/(k-1)+\lambda$ and this expression is not integral. Thus, a player in row $1$ can only meet a player in row $k$ again after at least $2\frac{n-\ell}{k(k-1)}$ rounds.

Since $2\frac{n-\ell}{k(k-1)}\geq\frac{n}{k(k-1)}$ if $n\geq \frac{2k}{k-2}$, the condition is satisfied for $n$ sufficiently large.

Similarly, we have to check that two players in two rows with a relative distance of at most $k-2$ do not meet more than once. Since $\frac{n-\ell}{k(k-2)}\geq\frac{n}{k(k-1)}$ if $n\geq k^2-k$, the condition is also satisfied for $n$ sufficiently large.

Observe that there are infinitely many $n$ and $\ell$ such that $\ell=\frac{n}{k}+1$, $n$ is divisible by $k(k-1)$ and $\ell \equiv k \mod{k(k-1)}$ and thus the result follows for sufficiently large $n$.
\end{proof}

We turn our attention to the problem of characterizing tournaments that are not extendable after $\lfloor\frac{n}{k(k-1)}\rfloor$ rounds. Assuming the Equitable $\Delta$-Coloring Conjecture (E$\Delta$CC) is true, we give an exact characterization of the feasibility graphs of tournaments that cannot be extended after $\lfloor\frac{n}{k(k-1)}\rfloor$ rounds. The existence of an instance not fulfilling these conditions would immediately disprove the E$\Delta$CC.
Furthermore, this characterization allows us to guarantee  $\lfloor\frac{n}{k(k-1)}\rfloor+1$ rounds in every tournament when the last two rounds are chosen carefully. 

\subsubsection*{Equitable $\Delta$-Coloring Conjecture \citep{chen1994equitable}.}
Let $G$ be a connected graph with maximum degree $\Delta(G) \leq \ell$. Then $G$ is not equitable $\ell$-colorable if and only if one of the following three cases occurs:
\begin{enumerate}
\item[(i)] $G=K_{\ell+1}$.
\item[(ii)] $\ell=2$ and $G$ is an odd cycle.
\item[(iii)] $\ell$ odd and $G=K_{\ell,\ell}$.
\end{enumerate}
\subsubsection*{}The conjecture was first stated by \citet{chen1994equitable} and is proven for $|V|=k\cdot\ell$ and $k=2,3,4$. See the Chen-Lih-Wu theorem for $k=2$ and \citet{kierstead2015refinement} for $k=3,4$. Both results make use of Brooks' theorem \citep{brooks1941coloring}. For $k>4$, the conjecture is still open.

\begin{proposition}
If \EDCC is true, a tournament with $\lfloor\frac{n}{k(k-1)} \rfloor$ rounds cannot be extended if and only if $K_{\frac{n}{k}+1}$ is a subgraph of the complement graph $\bar{G}$.
\label{prop:charOneRoundMore}
\end{proposition}

Before we start the proof we state the following claim, which we will need in the proof.

\begin{claim}
\label{claim:connectedcomponents}
Let $G$ be a graph with $|G|$ vertices and let $m$ be such that $|G| \equiv 0 \mod{m}$. Given an equitable $m$-coloring for every connected component $G_i$ of $G$, there is an equitable $m$-coloring for $G$.
\end{claim}

\begin{proof}
Let $G$ consist of connected components $G_1, \dots, G_c$. In every connected component $G_i$, $i \in \{1, \dots, c\}$, there are $\ell_i \equiv |G_i| \mod{m}$ \emph{large color classes}, i.e., color classes with $\lfloor\frac{|G_i|}{m}\rfloor +1$ vertices and $m-\ell_i$ \emph{small color classes}, i.e., color classes with $\lfloor\frac{|G_i|}{m}\rfloor$ vertices. First note that from $\ell_i \equiv|G_i| \mod{m}$, it follows that $(\sum \ell_i) \equiv (\sum |G_i|) \equiv |G| \equiv 0 \mod{m}$, i.e., $\sum \ell_i$ is divisible by $m$.

In the remainder of the proof, we will argue that we can recolor the color classes in the connected components to new color classes $\{1,\ldots, m\}$ that form an equitable coloring. The proof is inspired by McNaughton's wrap around rule \cite{mcnaughton1959scheduling} from Scheduling. Pick some connected component $G_i$ and assign the $\ell_i$ large color classes to the new colors $\{1, \dots, \ell_i\}$. Choose some next connected component $G_j$ with $j\neq i$ and assign the $\ell_j$ large color classes to the new colors $\{(\ell_{i} + 1) \mod m, \dots, (\ell_{i} + \ell_j) \mod{m} \}$. Proceed analogously with the remaining connected components. Note that $\ell_i<m$ for all $i \in \{1, \dots, c\}$, thus we assign at most one large color class from each component to every new color. Finally, for each connected component we add a small color class to all new colors if no large color class from this connected component was added in the described procedure.

Each new color class contains exactly $\frac{\sum \ell_i}{m}$ large color classes and $c - \frac{\sum \ell_i}{m}$ small color classes and has thus the same number of vertices.
This gives us an $m$-equitable coloring of $G$.
\end{proof}

\begin{proof}[Proof of Proposition~\ref{prop:charOneRoundMore}.]
First, assume that $K_{\frac{n}{k}+1}$ is a subgraph of $\bar{G}$, the complement of the feasibility graph. We will show that the tournament is not extendable. To construct an additional round, on the one hand at most one vertex from the complete graph $K_{\frac{n}{k}+1}$ can be in each clique. On the other hand, there are only $\frac{n}{k}$ cliques in a round which directly implies that the tournament is not extendable.

Second, assume that the tournament cannot be extended after $\lfloor \frac{n}{k(k-1)} \rfloor$ rounds. We will show that given \EDCC $K_{\frac{n}{k}+1}$ is a subgraph of $\bar{G}$. If there is an equitable $\frac{n}{k}$-coloring for every connected component, then by \Cref{claim:connectedcomponents} there is an equitable $\frac{n}{k}$-coloring of $\bar{G}$ and thus also a next round of the tournament. This is a contradiction to the assumption. Thus, there is a connected component $\bar{G}_i$ that is not equitable $\frac{n}{k}$-colorable.

After $\lfloor \frac{n}{k(k-1)} \rfloor$ rounds, the degree of all vertices $v$ in $\bar{G}_i$ is $\Delta(\bar{G})=(k-1)\lfloor\frac{n}{k(k-1)}\rfloor \leq \frac{n}{k}$. By E$\Delta$CC for $\ell =\frac{n}{k}$, one of the following three cases occur: (i) $\bar{G}_i=K_{\frac{n}{k} + 1}$, or (ii) $\frac{n}{k}=2$ and $\bar{G}_i$ is an odd cycle, or (iii) $\frac{n}{k}$ is odd and $\bar{G}_i=K_{\frac{n}{k}, \frac{n}{k}}$. We will show that (ii) and (iii) cannot occur, thus $\bar{G}_i=K_{\frac{n}{k} + 1}$, which will finish the proof.

Assume that (ii) occurs, i.e., $\frac{n}{k}=2$ and $\bar{G}_i$ is an odd cycle. Since we assume $k\geq 3$ in this section, an odd cycle can only be formed from a union of complete graphs $K_k$ if there is only one round with $k=3$. Thus, we have that $n=6$. In this case, (ii) reduces to (i) because $\bar{G}_i = K_{3} = K_{\frac{n}{k}+1}$.

Next, assume that (iii) occurs, i.e., $\frac{n}{k}$ is odd and $\bar{G}_i=K_{\frac{n}{k}, \frac{n}{k}}$. Given that $k \geq 3$, we will derive a contradiction. Since $k\geq 3$, every clique of size $k$ contains an odd cycle. This implies $\bar{G}_i$ contains an odd cycle, contradicting that $\bar{G}_i$ is bipartite.
\end{proof}

Note that any tournament with $H=K_k$ and $\lfloor \frac{n}{k(k-1)}\rfloor$ rounds which does not satisfy the condition in \Cref{prop:charOneRoundMore} would disprove the \EDCC.

\begin{proposition}
Let $n>k(k-1)$. If \EDCC is true, then by choosing round $\lfloor \frac{n}{k(k-1)}\rfloor$ carefully, there always exists a tournament with $\lfloor \frac{n}{k(k-1)}\rfloor + 1$ rounds, 
\end{proposition}

\begin{proof}
A tournament with $\lfloor \frac{n}{k(k-1)}\rfloor$ rounds is either extendable or by \Cref{prop:charOneRoundMore}, at least one connected component of the complement of the feasibility graph is equal to $K_{\frac{n}{k}+1}$. In the former case, we are done. So assume the latter case.  Denote the connected components that are equal to $K_{\frac{n}{k}+1}$
by $\bar{G}_1, \ldots , \bar{G}_c$. First, we shorten the tournament by eliminating the last round and then extend it by two other rounds.

First of all notice that to end up with $\bar{G}_i=K_{\frac{n}{k}+1}$ after $\lfloor \frac{n}{k(k-1)}\rfloor$ rounds all matches between vertices in $\bar{G}_i$ need to be scheduled entirely inside $\bar{G}_i$. The reason for this is that $\bar{G}_i$ has $\frac{n^2}{2k^2}+ \frac{n}{2k}$ edges which is the maximum number of edges that can arise from $\lfloor \frac{n}{k(k-1)}\rfloor$ rounds with $\frac{n}{k}+1$ players.

Clearly, the last round of the original tournament corresponds to a $K_k$-factor in the feasibility graph of the shortened tournament. By the assumed structure of the feasibility graph, all cliques $K_k$ are either completely within $\bar{G}_i$, $i \in \{1,\dots,c\}$ or completely within $V \setminus \bigcup_{i \in \{ 1, \ldots , c\}} \bar{G}_i$. Thus, for each $i \in \{1,\dots,c\}$, all edges between $\bar{G}_i$ and $V \setminus \bar{G}_i$ are not present in the complement of the feasibility graph.

If $c=1$, select a vertex $v_1 \in \bar{G}_1$ and $v_2 \in V \setminus \bar{G}_1$. Exchange these vertices to get a $K_k$-factor with which the shortened tournament is extended. More precisely, $v_1$ is paired with the former clique of $v_2$ and vice versa, while all remaining cliques stay the same.
Since $k<\frac{n}{k}+1$ by assumption, this ensures that there is no set of $\frac{n}{k}+1$ vertices for which we have only scheduled matches within this group. Thus, after extending the tournament, no connected component in the complement of the feasibility graph corresponds to $K_{\frac{n}{k}+1}$.
By \Cref{prop:charOneRoundMore}, the tournament can be extended to have $\lfloor \frac{n}{k(k-1)}\rfloor + 1$ rounds.

If $c>1$, we select a vertex $v_i$ from each $\bar{G}_i$ for $i \in \{ 1,\ldots , c\}$.
We exchange the vertices in a cycle to form new cliques, i.e., $v_i$ is now paired with the vertices in the old clique of $v_{i+1}$ for all $i \in \{1, \dots, c\}$, where $v_{c+1}=v_1$. By adding this new $K_k$-factor, we again ensure that there is no set of $\frac{n}{k}+1$ vertices for which we have only scheduled matches within this group. By applying \Cref{prop:charOneRoundMore} we can extend the tournament for another round, which finishes the proof.
\end{proof}

\section{The Greedy Oberwolfach Problem} \label{sec:gop}

In this section we consider tournaments with $H=C_k$ for $k \geq 3$. Dependent on the number of participants $n$ and $k$ we derive bounds on the number of rounds that can be scheduled greedily in such a tournament.

Before we continue with the theorem, we first state two graph-theoretic results and a conjecture.

\subsubsection*{Aigner-Brandt Theorem \citep{aigner1993embedding}.}
Let $G$ be a graph with minimum degree $\delta(G) \geq \frac{2n-1}{3}$. Then $G$ contains any graph $H$ with at most $n$ vertices and maximum degree $\Delta(H)= 2$ as a subgraph.

\subsubsection*{Alon-Yuster Theorem \citep{alon1996h}.}
For every $\epsilon>0$ and for every $k\in\mathbb{N}$, there exists an $n_0=n_0(\epsilon,k)$ such that for every graph $H$ with $k$ vertices and for every $n>n_0$, any graph $G$ with $nk$ vertices and minimum degree $\delta(G)\geq \left(\frac{\chi(H)-1}{\chi(H)}+\epsilon\right)nk$ has an $H$-factor that can be computed in polynomial time. Here, $\chi(H)$ denotes the chromatic number of $H$, i.e., the smallest possible number of colors for a vertex coloring of $H$.

\subsubsection*{El-Zahar's Conjecture \citep{el1984circuits}.}
Let $G$ be a graph with $n=k_1+\ldots+k_{\ell}$. If $\delta(G)\geq \lceil \frac{1}{2} k_1 \rceil + \ldots + \lceil \frac{1}{2} k_\ell \rceil$, then G contains $\ell$ vertex disjoint cycles of lengths $k_1, \ldots, k_\ell$.

El-Zahar's Conjecture is proven to be true for $k_1=\ldots=k_{\ell}=3$ \citep{corradi1963maximal}, $k_1=\ldots=k_{\ell}=4$ \citep{wang2010proof}, and $k_1=\ldots=k_{\ell}=5$ \citep{wang2012disjoint}.

\begin{theorem}\label{thm:obe1}
Let $H=C_k$, Algorithm~\ref{algo:greedy} outputs a tournament with at least
\begin{enumerate}
    \item $\lfloor\frac{n+4}{6}\rfloor$ rounds for all $n \in \mathbb{N}$\;,
    \item $\lfloor\frac{n+2}{4}-\epsilon\cdot n \rfloor$ rounds for $n$ large enough, $k$ even and for fixed $\epsilon>0$\;,
\end{enumerate}
If El-Zahar's conjecture is true, the number of rounds improves to $\lfloor\frac{n+2}{4}\rfloor$ for $k$ even and  $\lfloor\frac{n+2}{4}-\frac{n}{4k}\rfloor$ for $k$ odd and all $n \in \mathbb{N}$.
\end{theorem}

\begin{proof}
\textbf{Statement 1.} Recall that Algorithm \ref{algo:greedy} starts with the empty tournament and the corresponding feasibility graph is the complete graph, where the degree of every vertex is $n-1$. In each iteration of the algorithm, a $C_k$-factor is deleted from the feasibility graph and thus every vertex loses $2$ edges.
We observe that as long as the constructed tournament has at most $\lfloor\frac{n-2}{6}\rfloor$ rounds, the degree of every vertex in the feasibility graph is at least $n-1-\lfloor\frac{n-2}{3}\rfloor \geq \frac{2n-1}{3}$. 
Since a $C_k$-factor with $n$ vertices has degree $2$, by the Aigner-Brandt theorem $G$ contains a $C_k$-factor. It follows that the algorithms runs for another iteration. In total, the number of rounds of the tournament is at least $\lfloor\frac{n-2}{6}\rfloor+1 = \lfloor\frac{n+4}{6}\rfloor$.

\textbf{Statement 2.} Assume $k$ is even. We have that the chromatic number $\chi(C_k)=2$. As long as Algorithm~\ref{algo:greedy} runs for at most $\lfloor\frac{n-2}{4}-\epsilon\cdot n\rfloor$ iterations, the degree of every vertex in the feasibility graph is at least $n-1-2 \cdot \lfloor\frac{n-2}{4}-\epsilon\cdot n\rfloor\geq n-1-\frac{n-2}{2}+2\epsilon\cdot n = \frac{n}{2}+2\epsilon\cdot n$. Hence by the Alon-Yuster theorem with $k'=k$, $n'=\frac{n}{k}$ and $\epsilon'=2\epsilon$, a $C_k$-factor exists for $n$ large enough and thus another iteration is possible. This implies that Algorithm~\ref{algo:greedy} is guaranteed to construct a tournament with $\lfloor\frac{n-2}{4}-\epsilon\cdot n\rfloor + 1$ rounds.

\textbf{Statement El-Zahar, $k$ even.} As long as Algorithm~\ref{algo:greedy} runs for at most $\lfloor\frac{n-2}{4}\rfloor$ iterations, the degree of every vertex in the feasibility graph is at least $n-1-2 \cdot \lfloor\frac{n-2}{4}\rfloor\geq n-1-\frac{n-2}{2} = \frac{n}{2}$. Hence from El-Zahar's conjecture with $k_1 = k_2 = \dots = k_\ell = k$ and $\ell=\frac{n}{k}$, we can deduce that a $C_k$-factor exists as $\frac{k}{2}\cdot \frac{n}{k}=\frac{n}{2}$, and thus another iteration is possible. This implies that Algorithm~\ref{algo:greedy} is guaranteed to construct a tournament with $\lfloor\frac{n-2}{4}\rfloor + 1$ rounds.

\textbf{Statement El-Zahar, $k$ odd.} As long as Algorithm~\ref{algo:greedy} runs for at most $\lfloor\frac{n-2}{4}-\frac{n}{4k}\rfloor$ iterations, the degree of every vertex in the feasibility graph is at least $n-1-\frac{n-2}{2}+\frac{n}{2k} = \frac{n}{2}+ \frac{n}{2k}$. Hence from El-Zahar's conjecture with $k_1 = k_2 = \dots = k_\ell = k$ and $\ell=\frac{n}{k}$, we can deduce that a $C_k$-factor exists as $\frac{k+1}{2}\cdot \frac{n}{k}=\frac{n}{2}+ \frac{n}{2k}$, and thus the constructed tournament can be extended by one more round. This implies that the algorithm outputs a tournament with at least $\lfloor\frac{n-2}{4}-\frac{n}{4k}\rfloor +1$ rounds.
\end{proof}

\begin{proposition}\label{prop:obe2}
Let $H=C_k$ for fixed $k$. Algorithm~\ref{algo:greedy} can be implemented such that it runs in polynomial time for at least
\begin{enumerate}
    \item $\lfloor\frac{n+2}{4}-\epsilon\cdot n \rfloor$ rounds, for $k$ even, and fixed $\epsilon>0$, or stops if, in the case of small $n$, no additional round is possible\;,
    \item $\lfloor\frac{n+3}{6}-\epsilon \cdot n\rfloor$ rounds, for $k$ odd, and fixed $\epsilon>0$\;.
\end{enumerate}
\end{proposition}

\begin{proof}
\textbf{Case 1.} Assume $k$ is even. By the Alon-Yuster theorem, analogously to Case 2 of Theorem~\ref{thm:obe1}, the first $\lfloor\frac{n+2}{4}-\epsilon\cdot n \rfloor$ rounds exist and can be computed in polynomial time, given that $n>n_0$ for some $n_0$ that depends on $\epsilon$ and $k$. Since $\epsilon$ and $k$ are assumed to be constant, $n_0$ is constant. By enumerating all possibilities in the case $n\leq n_0$, we can bound the running time for all $n \in \mathbb{N}$ by a polynomial function in $n$. Note that the Alon-Yuster theorem only implies existence of $\lfloor\frac{n+2}{4}-\epsilon \cdot n\rfloor$ rounds if $n>n_0$, so it might be that the algorithm stops earlier, but in polynomial time, for $n\leq n_0$.

\textbf{Case 2.} Assume $k$ is odd. First note that the existence of the first $\lfloor\frac{n-3}{6}-\epsilon\cdot n\rfloor \leq \lfloor \frac{n+4}{6} \rfloor$ rounds follows from Theorem~\ref{thm:obe1}. Observe that for odd cycles $C_k$ the chromatic number is $\chi(C_k)=3$. As long as Algorithm~\ref{algo:greedy} runs for at most $\lfloor\frac{n-3}{6}-\epsilon\cdot n\rfloor$ iterations, the degree of every vertex in the feasibility graph is at least $n-1-2 \cdot \lfloor\frac{n-3}{6}-\epsilon\cdot n\rfloor\geq n-1-\frac{n-3}{3}+2\epsilon\cdot n = \frac{2n}{3}+2\epsilon\cdot n$. Hence by the Alon-Yuster theorem with $\epsilon'=2\epsilon$, there is an $n_0$ dependent on $k$ and $\epsilon$ such that a $C_k$-factor can be computed in polynomial time for all $n>n_0$. Since $\epsilon$ and $k$ are assumed to be constant, $n_0$ is constant. By enumerating all possibilities for $n\leq n_0$ we can bound the running time of the algorithm by a polynomial function in $n$ for all $n \in \mathbb{N}$.
\end{proof}

\begin{corollary}
For any fixed $\epsilon>0$, Algorithm~\ref{algo:greedy} is a $\frac{1}{3+\epsilon} $-approximation algorithm for the Oberwolfach problem.
\end{corollary}

\begin{proof}
Fix $\epsilon > 0$.
\paragraph{Case 1} If $n \geq \frac{12}{\epsilon} +6$, we  choose $\epsilon '= \frac{1}{(3+ \epsilon) \frac{12}{\epsilon}}$ and use Proposition~\ref{prop:obe2} with $\epsilon'$. We observe
\begin{align*}&\left\lfloor\frac{n+3}{6}-\frac{1}{(3+ \epsilon) \frac{12}{\epsilon}}\cdot n\right\rfloor \cdot (3 + \epsilon) \geq  \left(\frac{n-3}{6}-\frac{1}{(3+ \epsilon) \frac{12}{\epsilon}}\cdot n\right) \cdot (3 + \epsilon) \\
= &\left(\frac{(n-3)(3 + \epsilon)}{6}-\frac{\epsilon}{12}\cdot n\right) = \left(\frac{n-3}{2} + \frac{(n-3) \epsilon}{6} -\frac{\epsilon}{12}\cdot n\right)\\ = &\left(\frac{n-3}{2} + \frac{(2 n \epsilon -6 \epsilon)}{12} -\frac{\epsilon n}{12}\right) = \frac{n-1}{2} + \frac{\epsilon( n  -6)-12}{12} \geq \frac{n-1}{2} \geq \text{OPT}\;.
\end{align*} 
\paragraph{Case 2} If $n < \frac{12}{\epsilon} +6$, 
$n$ is a constant and we can find a cycle-factor in each round by enumeration. By the Aigner-Brandt theorem the algorithm outputs $\lfloor \frac{n+4}{6}\rfloor \geq \frac{n-1}{6}$ rounds. This implies a $\frac{1}{3}> \frac{1}{3 + \epsilon}$ approximation algorithm.
\end{proof}

In the rest of the section, we show that the bound corresponding to El-Zahar's conjecture presented in Theorem \ref{thm:obe1} is essentially tight. Through a case distinction, we provide matching examples that show the tightness of the bounds provided by El-Zahar's conjecture for two of three cases. For $k$ even but not divisible by $4$, an additive gap of one round remains. All other cases are tight. Note that this implies that any improvement of the lower bound via an example by just one round (or by two for $k$ even but not divisible by $4$) would disprove El-Zahar's conjecture.

\begin{theorem}
There are infinitely many $n \in \N$ for which there exists a tournament with $H=C_k$ that is not extendable after
\begin{enumerate}
    \item $\lfloor\frac{n+2}{4}-\frac{n}{4k}\rfloor$ rounds if $k$ is odd\;,
    \item $\lfloor\frac{n+2}{4}\rfloor$ rounds if $k \equiv 0 \mod{4}$\;,
    \item $\lfloor\frac{n+2}{4}\rfloor+ 1$ rounds if $k \equiv 2 \mod{4}$\;.
\end{enumerate}
\end{theorem}

\begin{proof}
\textbf{Case 1.} Assume $k$ is odd. Let $n =2k\sum_{j=0}^i k^j$ for some integer $i\in\N$.
We construct a tournament with $n$ participants and $H=C_k$. To do so, we start with the empty tournament and partition the set of vertices of the feasibility graph into two disjoint sets $A$ and $B$. The sets are chosen such that $A \cup B = V$, and $|A| = \frac{n}{2}-\frac{n}{2k}+1= (k-1)\sum_{j=0}^i k^j+1=k^{i+1}$, $|B|= \frac{n}{2}+\frac{n}{2k}-1$ vertices. We observe that $|A|\leq|B|$, since $\frac{n}{2k}\geq 1$.
We construct a tournament such that in the feasibility graph all edges between vertices in $A$ are deleted. 
To do so, we use a result of \citet{alspach1989oberwolfach}, who showed that there is a solution for the Oberwolfach problem for all odd $k$ with $n \equiv 0 \mod{k}$ and $n$ odd.

Observe that $|A| \equiv 0 \mod{k}$, thus $|B| \equiv 0 \mod{k}$. Furthermore, $|A|-1$ is even and since $n$ is even this also applies to $\abs{B}-1$. By using the equivalence of the Oberwolfach problem to complete tournaments, there exists a complete tournament within $A$ and within $B$.
We combine these complete tournaments to a tournament for the whole graph with $\min\{|A|-1, |B|-1\}/2 = \frac{|A|-1}{2} = \frac{n}{4}-\frac{n}{4k}$ rounds. Since $|A|$ is odd, the number of rounds is integral.

Considering the feasibility graph of this tournament, there are no edges between vertices in $A$. Thus, every cycle of length $k$ can cover at most $\frac{k-1}{2}$ vertices of $A$. We conclude that there is no $C_k$-factor for the feasibility graph, since $\frac{n}{k}\cdot \frac{k-1}{2}=\frac{n}{2}-\frac{n}{2k}$, so we cannot cover all vertices of $A$. Thus, we constructed a tournament with $\frac{n}{4}-\frac{n}{4k}=\lfloor\frac{n+2}{4}-\frac{n}{4k}\rfloor$ rounds that cannot be extended.

\textbf{Case 2.} Assume $k$ is divisible by $4$. Let $n=i\cdot k$ for some odd integer $i\in\N$. We construct a tournament with $n$ participants by dividing the vertices of the feasibility graph into two disjoint sets $A$ and $B$ such that $|A| = |B|= \frac{n}{2} = i \cdot \frac{k}{2}$. \citet{liu2003equipartite} showed that there exist $n/4$ disjoint $C_k$-factors in a complete bipartite graph with $n/2$ vertices on each side of the bipartition, if $n/2$ is even. That is, every edge of the complete bipartite graph is in exactly one $C_k$-factor. 
Since $n/2$ is even by case distinction, there is a tournament with  $n/4 = \lfloor \frac{n+2}{4} \rfloor$ rounds such that in the feasibility graph there are only edges within $A$ and within $B$ left. Since $i$ is odd, $|A| = i \cdot \frac{k}{2}$ is not divisible by $k$. Thus, it is not possible to schedule another round by choosing only cycles within sets $A$ and $B$.

\textbf{Case 3.} Assume $k$ is even, but not divisible by 4. Let $n=i\cdot k$ for some odd integer $i\in\N_{\geq 9}$.
We construct a tournament with $n$ participants that is not extendable after  $\frac{n+2}{4} + 1$ rounds in two phases. First, we partition the vertices into two disjoint sets $A$ and $B$, each of size $\frac{n}{2}$, and we construct a base tournament with $\frac{n-2}{4}$ rounds such that in the feasibility graph only edges between sets $A$ and $B$ are deleted. Second, we extend the tournament by two additional carefully chosen rounds.
After the base tournament, the feasibility graph consists of two complete graphs $A$ and $B$ connected by a perfect matching between all vertices from $A$ and all vertices from $B$. We use the additional two rounds to delete all of the matching-edges except for one. Using this, we show that the tournament cannot be extended.

In order to construct the base tournament, we first use a result of \citet{alspach1989oberwolfach}. It states that there always exists a solution for the Oberwolfach problem with $n'$ participants and cycle length $k'$ if $n'$ and $k'$ are odd and $n' \equiv 0 \mod{k'}$.

We choose $n'=n/2$ and $k' = k/2$ (observe that by assumption $k\geq 6$ and thus $k'\geq 3$) and then apply the result by \citet{alspach1989oberwolfach} to obtain a solution for the Oberwolfach problem with $n'$ and $k'$. Next we use a construction relying on an idea by \citet{archdeacon2004cycle} to connect two copies of the Oberwolfach solution. Fix the solution for the Oberwolfach problem with $n/2$ participants and cycle length $\frac{k}{2}$, and apply this solution to $A$ and $B$ separately.
Consider one round of the tournament and denote the $C_{\frac{k}{2}}$-factor in $A$ by $(a_{1+j}, a_{2+j}, \dots, a_{\frac{k}{2}+j})$ for $j=0, \frac{k}{2}, k, \dots, \frac{n}{2}-\frac{k}{2}$. By symmetry, the $C_{\frac{k}{2}}$-factor in $B$ can be denoted by $(b_{1+j}, b_{2+j}, \dots, b_{\frac{k}{2}+j})$ for $j=0, \frac{k}{2}, k, \dots, \frac{n}{2}-\frac{k}{2}$. We design a $C_k$-factor in the feasibility graph of the original tournament.
 For each $j \in \{0, \frac{k}{2}, k, \dots, \frac{n}{2}-\frac{k}{2}\}$, we construct a cycle $(a_{1+j},b_{2+j},a_{3+j}, \dots, a_{\frac{k}{2}+j},b_{1+j}, a_{2+j},b_{3+j},\dots,b_{\frac{k}{2}+j})$ of length $k$ in $G$. These edges are not used in any other round due to the construction and we used the fact that $\frac{k}{2}$ is odd. We refer to \Cref{fig:basetournament} for an example of one cycle for $k=10$. Since each vertex is in one cycle in each round, the construction yields a feasible round of a tournament. Applying this procedure to all rounds yields the base tournament with $\frac{n-2}{4}$ rounds.
 
\begin{figure}[t]
\begin{minipage}{0.48\textwidth}
\centering
\begin{tikzpicture}[scale=0.7]
    \draw (0,0) ellipse (4cm and 1.1cm);
    \draw (0,-2.5) ellipse (4cm and 1.1cm);
   \node[circle,fill=black,inner sep=0pt,minimum size=5pt] at (-2,0){};
        \node[left] at (-2,0) {$a_1$};
    \node[circle,fill=black,inner sep=0pt,minimum size=5pt] at (-1,-0.5){};
        \node[left] at (-1,-0.5) {$a_2$};
    \node[circle,fill=black,inner sep=0pt,minimum size=5pt] at (1,-0.5){};
           \node[right] at (1,-0.5) {$a_3$};
    \node[circle,fill=black,inner sep=0pt,minimum size=5pt] at (2,0){};
           \node[right] at (2,0) {$a_4$};
   \node[circle,fill=black,inner sep=0pt,minimum size=5pt] at (0,0.5){};
          \node[above] at (0,0.5) {$a_5$};
   
       \node[circle,fill=black,inner sep=0pt,minimum size=5pt] at (-2,-2.5){};
        \node[left] at (-2,-2.5) {$b_1$};
    \node[circle,fill=black,inner sep=0pt,minimum size=5pt] at (-1,-3){};
     \node[left] at (-1,-3) {$b_2$};
    \node[circle,fill=black,inner sep=0pt,minimum size=5pt] at (1,-3){};
         \node[right] at (1,-3) {$b_3$};
    \node[circle,fill=black,inner sep=0pt,minimum size=5pt] at (2,-2.5){};
         \node[right] at (2,-2.5) {$b_4$};
   \node[circle,fill=black,inner sep=0pt,minimum size=5pt] at (0,-2){};
         \node[above] at (0,-2) {$b_5$};

    \draw (-2,0) -- (-1,-0.5);
    \draw (-1,-0.5) -- (1,-0.5);
    \draw (1,-0.5) -- (2,0);
    \draw (2,0) -- (0,0.5);
    \draw (0,0.5) -- (-2,0);
    
        \draw (-2,-2.5) -- (-1,-3);
    \draw (-1,-3) -- (1,-3);
    \draw (1,-3) -- (2,-2.5);
    \draw (2,-2.5) -- (0,-2);
    \draw (0,-2) -- (-2,-2.5);
    \node at (0,-4){};

\end{tikzpicture}
\end{minipage}
\begin{minipage}{0.48\textwidth}
\centering
\begin{tikzpicture}[scale=0.7]
    \draw (0,0) ellipse (4cm and 1.1cm);
    \draw (0,-2.5) ellipse (4cm and 1.1cm);
   \node[circle,fill=black,inner sep=0pt,minimum size=5pt] at (-2,0){};
        \node[left] at (-2,0) {$a_1$};
    \node[circle,fill=black,inner sep=0pt,minimum size=5pt] at (-1,-0.5){};
        \node[above] at (-1,-0.5) {$a_2$};
    \node[circle,fill=black,inner sep=0pt,minimum size=5pt] at (1,-0.5){};
           \node[above] at (1,-0.5) {$a_3$};
    \node[circle,fill=black,inner sep=0pt,minimum size=5pt] at (2,0){};
           \node[right] at (2,0) {$a_4$};
   \node[circle,fill=black,inner sep=0pt,minimum size=5pt] at (0,0.5){};
          \node[left] at (0,0.5) {$a_5$};
   
       \node[circle,fill=black,inner sep=0pt,minimum size=5pt] at (-2,-2.5){};
        \node[left] at (-2,-2.5) {$b_1$};
    \node[circle,fill=black,inner sep=0pt,minimum size=5pt] at (-1,-3){};
     \node[left] at (-1,-3) {$b_2$};
    \node[circle,fill=black,inner sep=0pt,minimum size=5pt] at (1,-3){};
         \node[left] at (1,-3) {$b_3$};
    \node[circle,fill=black,inner sep=0pt,minimum size=5pt] at (2,-2.5){};
         \node[right] at (2,-2.5) {$b_4$};
   \node[circle,fill=black,inner sep=0pt,minimum size=5pt] at (0,-2){};
         \node[below] at (0,-2) {$b_5$};
    \draw (-2,0) -- (-1,-3);
    \draw (-1,-3) -- (1,-0.5);
    \draw (1,-0.5) -- (2,-2.5);
    \draw (2,-2.5) -- (0,0.5);
    \draw (0,0.5) -- (-2,-2.5);
     \draw (-2,-2.5) -- (-1,-0.5);
    \draw (-1,-0.5) -- (1,-3);
    \draw (1,-3) -- (2,0);
    \draw (2,0) -- (0,-2);
    \draw (0,-2) -- (-2,0);
    \node at (0,-4){};

\end{tikzpicture}
\end{minipage}
\caption{Construction of the base tournament. We transform two cycles of length $5$ into one cycle of length $10$.}\label{fig:basetournament}
\end{figure}
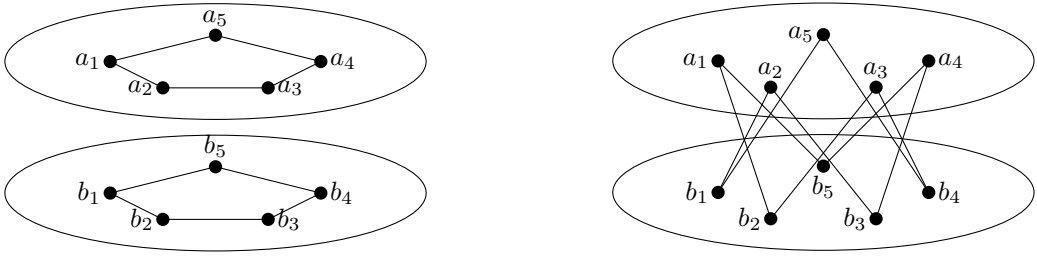
 
For each edge $e=\{a_{\bar{j}},a_j\}$ with $j\neq \bar{j}$ which is deleted in the feasibility graph of the tournament within $A$, we delete the edges $\{a_{\bar{j}},b_j\}$ and $\{a_j,b_{\bar{j}}\}$ in the feasibility graph. After the base tournament, all edges between $A$ and $B$ except for the edges $(a_1,b_1), (a_2,b_2), \dots , (a_{\frac{n}{2}},b_\frac{n}{2})$ are deleted in the feasibility graph.

In the rest of the proof, we extend the base tournament by two additional rounds. These two rounds are designed in such a way that after the rounds there is exactly one edge connecting a vertex from $A$ with one from $B$. To extend the base tournament by one round construct the cycles of the $C_k$-factor in the following way.  For $j\in\{0, \frac{k}{2}, k, \dots, \frac{n}{2}-\frac{k}{2}\}$, we construct the cycle $(a_{1+j},b_{1+j},b_{2+j},a_{2+j}, \ldots,b_{\frac{k}{2}-2 +j} b_{\frac{k}{2}-1 +j}, b_{\frac{k}{2} +j}, a_{\frac{k}{2} +j},a_{\frac{k}{2}-1 +j})$, see \Cref{fig:extendround1}. Since all edges within $A$ and $B$ are part of the feasibility graph as well as all edges $(a_{j'},b_{j'})$ for $j' \in \{ 1, \ldots , \frac{n}{2}\}$ this is a feasible construction of a $C_k$-factor and thus an extension of the base tournament.

\begin{figure}[t]
\centering
\begin{tikzpicture}[scale=0.9]
    \draw (0,0) ellipse (4cm and 1cm);
    \draw (0,-2.5) ellipse (4cm and 1cm);
    \node[circle,fill=black,inner sep=0pt,minimum size=5pt] at (-2,0){};
        \node[left] at (-2,0) {$a_1$};
    \node[circle,fill=black,inner sep=0pt,minimum size=5pt] at (-1,-0.5){};
        \node[left] at (-1,-0.5) {$a_2$};
    \node[circle,fill=black,inner sep=0pt,minimum size=5pt] at (1,-0.5){};
           \node[right] at (1,-0.5) {$a_3$};
    \node[circle,fill=black,inner sep=0pt,minimum size=5pt] at (2,0){};
           \node[right] at (2,0) {$a_4$};
   \node[circle,fill=black,inner sep=0pt,minimum size=5pt] at (0,0.5){};
          \node[left] at (0,0.5) {$a_5$};
   
       \node[circle,fill=black,inner sep=0pt,minimum size=5pt] at (-2,-2.5){};
        \node[left] at (-2,-2.5) {$b_1$};
    \node[circle,fill=black,inner sep=0pt,minimum size=5pt] at (-1,-3){};
     \node[left] at (-1,-3) {$b_2$};
    \node[circle,fill=black,inner sep=0pt,minimum size=5pt] at (1,-3){};
         \node[left] at (1,-3) {$b_3$};
    \node[circle,fill=black,inner sep=0pt,minimum size=5pt] at (2,-2.5){};
         \node[right] at (2,-2.5) {$b_4$};
   \node[circle,fill=black,inner sep=0pt,minimum size=5pt] at (0,-2){};
         \node[left] at (0,-2) {$b_5$};
    
    \draw (-2,0) -- (-2,-2.5);
    \draw (-2,-2.5) -- (-1,-3);
    \draw (-1,-3) -- (-1,-0.5);
    \draw (-1,-0.5) -- (1,-0.5);
    \draw (1,-0.5) -- (1,-3);
    
        \draw (1,-3) -- (2,-2.5);
    \draw (2,-2.5) -- (0,-2);
    \draw (0,-2) -- (0,0.5);
    \draw (0,0.5) -- (2,0);
    \draw (2,0) -- (-2,0);
\end{tikzpicture}

\caption{An example of one cycle in the construction that is used for the extension of the base tournament.}\label{fig:extendround1}
\end{figure}
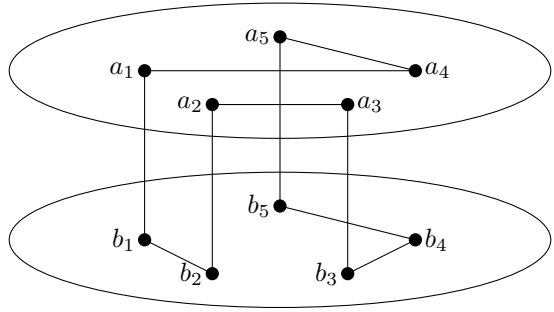

After the extension of the base tournament by one round the feasibility graph has the following structure. The degree of all vertices equals $\frac{n}{2}-2$ and the only edges between vertices from $A$ and $B$ are
\[\left\{(a_{\frac{k}{2}-1+j}, b_{\frac{k}{2}-1+j}) \mid j \in \left\{0, \frac{k}{2}, k, \dots, \frac{n}{2} - \frac{k}{2}\right\}\right\}\;.\]
We will construct one more round such that
after this round, there is only one of the matching edges remaining in the feasibility graph.

In order to do so, we will construct the $C_k$-factor with cycles $(C_1,\ldots, C_\frac{n}{k})$ by a greedy procedure as follows. Cycles $C_1, \dots, C_{\frac{n}{2k}-\frac{1}{2}}$ will all contain two matching edges and the other cycles none. In order to simplify notation we set
\[A_M = \left\{a_{\frac{k}{2}-1+j} \mid j \in \left\{0, \frac{k}{2}, k, \dots, \frac{n}{2} - \frac{k}{2}\right\}\right\}\;,\]
and $A_{-M} = A \setminus A_M$. We have $|A_{-M}| = \frac{n}{2}-\frac{n}{k}$. We define $B_M$ and $B_{-M}$ analogously. For some cycle $C_{z}$, $z\leq\frac{n}{2k}-\frac{1}{2}$, we greedily pick two of the matching edges. Let $(a_{\ell},b_{\ell})$ and $(a_r,b_r)$ be these two matching edges. To complete the cycle, we show that we can always construct a path from $a_{\ell}$ to $a_r$ by picking vertices from $A_{-M}$ and from $b_{\ell}$ to $b_r$ by vertices from $B_{-M}$. Assuming that we have already constructed cycles $C_1,\ldots, C_{z-1}$, there are still
\begin{align*}
    \frac{n}{2} - \frac{n}{k} - (z-1) \left(\frac{k}{2}-2\right)
\end{align*}
unused vertices in the set $A_{-M}$. Even after choosing some vertices for cycle $z$ the number of unused vertices in $A_{-M}$ is at least 
\[\frac{n}{2} - \frac{n}{k} - z \left(\frac{k}{2}-2\right) \geq \frac{n}{2} - \frac{n}{k} - z \frac{k}{2} \geq \frac{n}{2} - \frac{n}{k} - \frac{n}{2k} \frac{k}{2} = \frac{n}{4} - \frac{n}{k} \geq \frac{n}{12}\;.\]
Let $N(v)$ denote the neighborhood of vertex $v$. The greedy procedure that constructs a path from $a_{\ell}$ to $a_r$ works as follows. We set vertex $a_{\ell}$ active. For each active vertex $v$, we pick one of the vertices $a \in N(v) \cap A_{-M}$, delete $a$ from $A_{-M}$ and set $a$ active. We repeat this until we have chosen $\frac{k}{2}-3$ vertices. Next, we pick a vertex in $N(v) \cap A_{-M} \cap N(a_r)$ in order to ensure that the path ends at $a_r$. Since $|A_{-M}| \geq \frac{n}{12}$, we observe
\[|N(v) \cap A_{-M} \cap N(a_r)| \geq \frac{n}{12} - 1-2\;,\]  
so there is always a suitable vertex as $n\geq 9 k\geq 54$.
The construction for the path from $b_{\ell}$ to $b_r$ is analogous.

For cycles $C_{\frac{n}{2k}+\frac{1}{2}}, \dots, C_\frac{n}{k}$, there are still $\frac{n}{4}+\frac{k}{4}$ leftover vertices within $A$ and within $B$.
The degree of each vertex within the set of remaining vertices is at least $\frac{n}{4}+\frac{k}{4}-3$. This is large enough to apply the Aigner-Brandt theorem as $i\geq 9$ and $k\geq 6$.
In this way, we construct a $C_k$-factor in the feasibility graph. This means we can extend the tournament by one more round. In total we constructed a tournament of $\frac{n+2}{4}+1$ rounds, which is obviously equal to  $\lfloor \frac{n+2}{4} \rfloor +1$.

To see that this tournament cannot be extended further, consider the feasibility graph. Most of the edges within $A$ and $B$ are still present, while between $A$ and $B$ there is only one edge left. This means a $C_k$-factor can only consist of cycles that are entirely in $A$ or in $B$. Since $\abs{A}=\abs{B}$ and the number of cycles $\frac{n}{k}=i$ is odd, there is no $C_k$-factor in the feasibility graph and thus the constructed tournament is not extendable.
\end{proof}

\section{Conclusion and Outlook}
In this work, we studied the social golfer problem and the Oberwolfach problem from an optimization perspective. We presented bounds on the number of rounds that can be guaranteed by a greedy algorithm. 
For the social golfer problem the provided bounds are tight. Assuming El-Zahar's conjecture \citep{el1984circuits} holds, a gap of one remains for the Oberwolfach problem. This gives a performance guarantee for the optimization variant of both problems. Since both a clique- and cycle-factor can be found in polynomial time for graphs with high degree, the greedy algorithm is a $\frac{k-1}{2k^2-3k-1}$-approximation algorithm for the social golfer problem and a $\frac{1}{3+\epsilon}$-approximation algorithm for any fixed $\epsilon>0$ for the Oberwolfach problem.

Given some tournament it would be interesting to analyze the complexity of deciding whether the tournament can be extended by an additional round. Proving \NP-hardness seems particularly complicated since one cannot use any regular graph for the reduction proof, but only graphs that are feasibility graphs of a tournament.

Lastly, the general idea of greedily deleting particular subgraphs $H$ from base graphs $G$ can also be applied to different choices of $G$ and $H$.

\section*{Acknowledgement}
This research started after supervising the Master's thesis of David Kuntz. We thank David for valuable discussions. 

\bibliographystyle{apalike}
\bibliography{roundrobin}

\end{document}